\documentclass[lettersize,journal]{IEEEtran}
\usepackage{amsmath,amsfonts}
\usepackage{algorithmic}
\usepackage{algorithm}
\usepackage{array}
\usepackage[caption=false,font=normalsize,labelfont=sf,textfont=sf]{subfig}
\usepackage{textcomp}
\usepackage{stfloats}
\usepackage{url}
\usepackage{verbatim}
\usepackage{graphicx}
\usepackage{cite}
\usepackage{amsthm}  
\usepackage{amsmath} 
\usepackage{siunitx}
\DeclareSIUnit{\degreeCA}{\SIUnitSymbolDegree\kern.1em CA}
\DeclareSIUnit\bar{bar}

\newtheorem{theorem}{Theorem}

\newtheorem{remark}{Remark}

\begin{document}

\title{Nonlinear Stochastic Model Predictive Control with Generative Uncertainty in Homogeneous Charge Compression Ignition 
}

\author{Xu Chen, Kevin Kluge, Maximilian Basler, Lorenz D\"{o}rschel, Heike Vallery
\thanks{Xu Chen, Kevin Kluge, Maximilian Basler, Lorenz D\"{o}rschel and Heike Vallery are with Institute of Automatic Control, RWTH Aachen University, Campus-Boulevard 30, 52074 Aachen (e-mail: x.chen@irt.rwth-aachen.de; k.kluge@irt.rwth-aachen.de; m.basler@irt.rwth-aachen.de; l.doerschel@irt.rwth-aachen.de; h.vallery@irt.rwth-aachen.de).} 
}

\maketitle

\begin{abstract}
This work addresses the challenge of ignition timing and load control in homogeneous charge compression ignition engines operating subject to uncertainty from complex combustion dynamics and external disturbances. To handle this issue, we propose a nonlinear stochastic model predictive control approach explicitly incorporating distributional information of uncertainties. Specifically, we integrate an uncertainty model learned from empirical residual data to capture realistic probabilistic characteristics and handle the nonlinear additive uncertainty propagation within the prediction horizon based on polynomial chaos expansion. Additionally, we introduce a novel cost function based on maximum mean discrepancy, enabling direct penalization of the discrepancy between predicted and desired distributions of combustion indicators. The simulation results demonstrate that our proposed method achieves over a \SI{28}{\percent} reduction on combustion phasing variation and more than a \SI{26}{\percent} improvement in load tracking accuracy compared to traditional nonlinear and Gaussian-based predictive control strategies. These findings indicate the effectiveness of explicitly modeling uncertainty distributions and highlight the advantages of distribution-level performance index in robust combustion control. 


\end{abstract}

\begin{IEEEkeywords}
Model predictive control, stochastic model predictive control, homogeneous charge compression ignition, generative model, polynomial chaos expansion,  maximum mean discrepancy.
\end{IEEEkeywords}

\section{Introduction}
\IEEEPARstart{H}{omogeneous } Charge Compression Ignition (HCCI) and its controlled variant, Gasoline Controlled Auto-Ignition (GCAI), are promising combustion modes that offer high thermal efficiency and reduce nitrogen oxide and particulate emissions \cite{zhao2007hcci}. Unlike conventional spark-ignition combustion, HCCI relies on the auto-ignition of a premixed fuel-air mixture without a direct ignition trigger. As a result, HCCI combustion is highly sensitive to the initial thermodynamic state, gas composition, and operating conditions \cite{yao2009progress}. \par
To facilitate stable auto-ignition in gasoline engines, Exhaust Gas Recirculation (EGR) serves as a common approach by increasing the specific heat capacity of the mixture. However, EGR also induces a strong cyclic coupling effect due to residual gas retention \cite{lu2005fundamental}. This coupling gives rise to cycle-to-cycle variability in combustion phasing, adversely affecting the combustion stability and efficiency. In particular, under low-load operation, auto-ignition becomes highly susceptible to stochastic uncertainties in residual gas composition and temperature \cite{saxena2013fundamental}. These challenges motivate the development of advanced control algorithms capable of mitigating the impact of uncertainties in the in-cylinder states. \par
In the field of HCCI combustion control, Model Predictive Control (MPC) has attracted considerable attention owing to its strengths in explicit constraint handling, multivariable control, and flexible cost function design \cite{qin1997overview}. Early work \cite{ravi2012model} demonstrated closed-loop multi-input control using a linearized model. However, such simplifications failed to capture the inherent nonlinearity and stochastic nature of HCCI. Subsequent studies extended MPC to multi-cylinder engines \cite{erlien2012multi, erlien2013multicylinder}, addressed combustion mode switching \cite{widd2011control}, and incorporated wall thermal dynamics via linear parameter-varying models \cite{widd2011physics}, to improve control accuracy and adaptability across varying operating conditions. Efforts to reduce computational burden included Laguerre-based model simplification \cite{ebrahimi2015model} and its real-time hardware implementation \cite{ebrahimi2018real}. More recently, data-driven models such as LSTM networks \cite{gordon2024introducing} and extreme learning machines \cite{janakiraman2015nonlinear} have been employed as prediction models to better capture system behavior, but they often lack robustness guarantees. Physics-based nonlinear models, such as those in \cite{nuss2019nonlinear, de2022nonlinear}, offer improved fidelity but are computationally intensive and rarely account for uncertainty. Although widely applied, most existing MPC strategies for HCCI do not explicitly incorporate uncertainty within the prediction horizon, which limits their robustness and reliability under real-world disturbances.\par

To address the limitations of nominal MPC, which often results in aggressive control actions or constraint violations in the presence of uncertainty, researchers have developed several extended frameworks. Robust MPC improves closed-loop robustness by accounting for worst-case uncertainty, but this often leads to overly conservative control strategies that degrade overall performance \cite{bemporad2007robust}. As an alternative, Stochastic MPC (SMPC) offers a principled balance between performance and robustness by incorporating probabilistic descriptions of uncertainty \cite{heirung2018stochastic}. Typically, SMPC handles constraints by formulating chance constraints, which require satisfaction with high probability rather than in a deterministic manner \cite{mesbah2016stochastic}. Such a framework has shown particular effectiveness in applications involving statistically characterizable uncertainty, including building climate control and energy management systems \cite{oldewurtel2013stochastic, hans2015scenario, zhang2018stochastic}. \par
In the context of nonlinear SMPC, modeling uncertainty and its propagation remains a great challenge. To tackle this issue, sampling-based approaches have been widely explored. Sequential Monte Carlo methods sample control sequences and propagate forward rollouts to predict future system behavior \cite{kantas2009sequential}, while Markov Chain Monte Carlo-based approaches construct a Markov chain over the control input space whose stationary distribution concentrates around regions with low costs \cite{maciejowski2007nmpc}. Other methods, such as \cite{weissel2009stochastic}, use Gaussian Mixture Models to approximate transition densities through repeated propagation and refitting. Similarly, Gaussian Processes have been employed to capture both system dynamics and predictive uncertainty \cite{bradford2018stochastic}.  \par
While these methods improve the tractability of nonlinear SMPC problems, they are often computationally expensive and difficult to generalize or integrate into real-time applications. To address this, some works attempt to simplify the problem through local linearization. For instance, \cite{atmaram2020successive} used successive linearization to generate linear sub-models, and \cite{gronauer2024reinforcement} assumed Gaussian uncertainty with diagonal covariance for tractability. Although such approximations accelerate optimization, they may fail to capture the full complexity of nonlinear uncertainty. Other works, such as \cite{lu2020soft}, address distributional robustness via ambiguity sets but primarily focus on linear systems. In contrast to these nonparametric or approximate methods, parametric modeling of uncertainty offers high accuracy, computational efficiency and integration flexibility. In particular, learning-based generative models provide a promising direction for efficiently capturing complex, high-dimensional uncertainty distributions \cite{ruthotto2021introduction}. However, their integration into SMPC frameworks, especially for tractable uncertainty propagation, remains relatively underexplored. \par

Among various uncertainty propagation techniques, Polynomial Chaos Expansion (PCE) has gained attention as a data-efficient and analytically tractable method for SMPC \cite{paulson2019efficient}. It represents random variables using a series of orthogonal polynomial basis functions, facilitating efficient computation and analytical approximation of output distributions \cite{shen2020polynomial}. Several studies have applied PCE to linear systems under uncertainty. For instance, \cite{ma2023recursive} addressed a discrete-time linear system with time-invariant parametric uncertainty and employed Galerkin projection to derive deterministic dynamics for the PCE coefficients. Similarly, \cite{bhusal2022stochastic} reformulated the dynamics of a multi-agent system within the PCE coefficient space, enabling efficient state propagation under known stochastic inputs. To address nonlinear SMPC, \cite{mesbah2014stochastic} and \cite{fagiano2012nonlinear} have explored non-intrusive methods, where uncertainty was propagated through the nonlinear dynamics using sampling-based polynomial regression. This approach retains flexibility and avoids altering the underlying model equations. In \cite{paulson2019efficient}, the authors extended PCE to continuous-time nonlinear systems by propagating the conditional mean and covariance through a linearized model. Building on this, \cite{paulson2018nonlinear} addressed additive uncertainty using a degenerate distribution assumption and introduced an analytical PCE method for computing constraint backoffs to ensure probabilistic constraint satisfaction. While prior PCE-based methods have shown effectiveness, particularly in handling multiplicative uncertainty, their extension to nonlinear systems subject to additive uncertainty remains limited. In SMPC formulations for such systems, additive disturbance enters each prediction step, resulting in a growing number of random inputs across the prediction horizon. This increasing dimensionality challenges the scalability of existing PCE techniques and hinders their use in predictive uncertainty propagation. \par

To address the challenges of uncertainty modeling and propagation in nonlinear SMPC, we propose a Generative PCE-MMD-based SMPC (GEM-SMPC) framework with the following three key contributions: First, we develop a parametric uncertainty model trained on measurement data, which captures the underlying conditional distribution of uncertainties and provides a structured uncertainty representation within the control formulation. 
Second, we adopt the PCE approach to achieve tractable and efficient uncertainty propagation across the prediction horizon in the presence of nonlinear dynamics and additive uncertainty. Furthermore, we present an approximation strategy that preserves scalability to mitigate the dimensionality challenge induced by additive uncertainty at each prediction step. Third, instead of a conventional quadratic cost function, we propose a novel Maximum Mean Discrepancy (MMD)-based objective that penalizes the discrepancy between the predicted and target state distributions. This formulation allows for a more expressive evaluation of closed-loop performance by incorporating higher-order statistical characteristics of the distributions. \par
The remainder of this work is organized as follows: Section 2 introduces the structure and training framework of the uncertainty model and presents the reformulation of the SMPC problem under additive uncertainty and nonlinear dynamics. It includes the PCE-based state propagation, control policy parameterization, chance constraints reformulation, and the definition of the performance index. Section 3 demonstrates the application of the proposed GEM-SMPC framework to HCCI combustion control and provides comparisons with several baseline MPC strategies. The evaluation focuses on combustion efficiency and load tracking accuracy as key performance metrics. Section 4 concludes with key takeaways and a brief discussion of future research work.

\section{Methodology}
In this section, we formulate the nonlinear SMPC problem and organize the discussion of the GEM-SMPC framework into three parts. The first part focuses on uncertainty modeling with a parameterized generative model to map samples from the trivial distribution to the target uncertainty distribution. The second part introduces a PCE-based approach for uncertainty propagation under nonlinear dynamics. The third part presents an MMD-based objective function to quantify the divergence between the sample distribution and the target reference.
\subsection{Problem Formulation}
We first introduce the system model assumption and define the corresponding SMPC problem. The system under consideration is a discrete-time nonlinear system subject to additive nonlinear uncertainty:
\begin{equation} \label{nonlinear_ss}
    x_{k+1} = f(x_k, u_k )+g(x_k, w_k),
\end{equation}
where $x \in \mathbb{R}^{n_x} $ is the system state, $u \in \mathbb{R}^{n_u} $ is the control input, and $w \in \mathbb{R}^{n_w} $ is the external disturbance. We introduce the disturbance $w$ as a latent variable with a known distribution. Specifically, we assume that $w$ follows a standard normal distribution, i.e. $w \sim \mathcal{N}\left ( 0,I_{n_w} \right )  $. The resulting uncertainty in \eqref{nonlinear_ss}, referred to as the residual, is defined as:
\begin{equation} \label{residual}
    y_k = g(x_k, w_k),
\end{equation}
where uncertainty $y$ might arise from unmodeled dynamics, inherent stochasticity of the process, and noise. In addition, we assume that the state constraints lie within a compact set $\mathcal{X}:=  \left \{ x: G_x x\leq g_x \right \} $, and the input constraints are within a compact set $\mathcal{U}:=  \left \{ u: G_u u\leq g_u \right \} $. \par
Based on this model, we formulate the finite-horizon discrete-time optimal control problem at time step $k$ as follows: 
\begin{equation} \label{smpc}
\begin{aligned}
     \min_{\mathbf{u}_k} \quad & J(x_k, \mathbf{u}_k) \\
\text{s. t.} \quad 
& x_{0|k} = x_k, \\
& {x}_{i+1|k} = f(x_{i|k}, u_{i|k})+g(x_{i|k}, w_{i|k}), \ \forall i = 0:N-1, \\
& P(x_{i|k}\in \mathcal{X} )\geq 1-\epsilon_1, \quad \forall i = 1: N, \\
& P(u_{i|k} \in \mathcal{U} )\geq 1-\epsilon_2, \quad \forall i = 0: N-1, \\
& w_{i|k}\sim \mathcal{N}(0, I_{n_w}) , \quad \forall i = 0: N-1, \\
\end{aligned}
\end{equation}
where $J$ denotes the objective function, which we specify later, $N$ is the length of the prediction horizon, and $ \mathbf{u}_k :=\left [ u_{0|k}^\top, \dots, u_{N-1|k}^\top  \right ] ^\top $ is a sequence of control inputs. The parameters $\epsilon_1$ and $\epsilon_2 \in [0,1) $ represent the permissible probabilities of constraint violations for states and inputs. Furthermore, the latent random variable $w$ is assumed to be independently and identically distributed, and the interaction between the disturbance and the system state is governed by the nonlinear function $g$. 

\subsection{Uncertainty Modeling}
This subsection describes the construction of the uncertainty model. The deterministic mapping
$g:\mathbb{R}^{n_x}\times  \mathbb{R}^{n_w}\to \mathbb{R}^{n_x}$ in \eqref{residual} defines the conditional probability $P\left ( Y |X,W \right )$ as a Dirac measure. Our goal is to identify the mapping from the latent variable $w$ to the uncertainty $y$, conditioned on the current system state $x$. Throughout this section, capital letters denote random variables, and lowercase letters refer to specific realizations.\par
To train a uncertainty model introduced above, we aim to minimize the discrepancy between the real residual distribution and the one generated by the model. In this work, we employ the Wasserstein distance to quantify the divergence between these two distributions \cite{panaretos2019statistical}. The $1$-Wasserstein distance takes the form:
\begin{equation} \label{W-distance}
 W_1(P_q,P_{g}) =\inf_{\Gamma \in \mathcal{P}(\mathbb{R}^{n_x}\times \mathbb{R}^{n_x} ) }  \int \mathrm {c}(\nu _1,\nu_2)\Gamma(\mathrm {d}\nu _1,  \mathrm {d}\nu _2  )    ,
\end{equation}
where $\mathcal{P}$ denotes the set of Borel probability measures on $\mathbb{R}^{n_x} \times \mathbb{R}^{n_x}$, such that each measure has marginal distributions $P_q$ and $P_g$, and the associated random variables are $1$-integrable. Furthermore, $\mathrm {c}$ represents the distance metric, whose common choice is the Euclidean distance. It should be clarified that $P_q$ denotes the distribution of the real residuals, which is the conditional distribution given the observed state $x$, and its probability density function is $p_y(y| x)$.  Meanwhile, $P_g$ represents the distribution of the generated residuals, which is also the conditional distribution given the observed state $x$, and we write its factorization as:
\begin{equation} \label{factorization_y}
p_g(\hat y |  x) = \int_{\mathbb{R}^{n_z} }p_g(\hat y | w, x) p _w(w)\mathrm {d} w  ,
\end{equation}
where we assume that the introduced latent variables $W$ and the observed state $X$ are independent. The process mapping from $W$ and $X$ to $\hat Y$ constitutes the generative process, and the corresponding conditional probability defines the decoder conditional on states. \par
Due to the difficulty in directly obtaining the posterior distribution, namely the distribution of the latent variable $W$ conditioned on the data $Y$ in generative models, we introduce an encoder $Q(W | Y)$ to construct an approximate posterior, and we require the marginal probability constraint to hold, meaning that $Q_w(W):=\mathbb{E}_{P_y}\left [ Q(W | Y) \right ]  $ needs to be equal to the latent prior distribution $P_w$. Here, $P_y$ is the marginal distribution of residuals. \par
Next, we introduce the following theorem, which provides an equivalent formulation of the Wasserstein loss in \eqref{W-distance}. The proof follows the approach outlined in \cite{tolstikhin2017wasserstein}, while extending it to the conditional probability setting. 
\begin{theorem} \label{w-equiv}
Consider the deterministic decoder $P_g(\hat Y| W, X)$ with the function $g$ in \eqref{residual}. The $1$-Wasserstein distance in \eqref{W-distance} is equivalent to 
\begin{equation} \label{w-loss}
W_1 (P_q,P_g) = \inf_{ Q\in \left \{ Q\mid Q_w = P_w \right \}  } \ \mathbb{E}_{P_q} \mathbb{E}_{Q(W|Y)} \left [ c\left (  Y,g(W,X)\right )  \right ],
\end{equation}
where $Q$ denotes the encoder distribution, whose marginal distribution over $W$ is $Q_w$, and $P_w$ is the introduced prior distribution.
\end{theorem}
\begin{proof}
First, we define the $1$-Wasserstein distance with respect to the distributions of $Y$, $\hat Y$ and $W$ conditioned on $X$ of the form:
\begin{equation}  \label{special_w_dis}
    W^\ast_1(P_q,P_g) = \inf_{\Gamma \in \mathcal{P}_{Y,\hat Y, W|X} }\mathbb{E}_{\Gamma} \left [ c(Y, \hat Y) \right ], 
\end{equation}
where $ \mathcal{P}_{Y,\hat Y, W|X}$ denotes the set of measures, subject to the constraint that their marginals correspond to the distributions of $Y$, $\hat Y$ and $W$ conditioned on $X$, respectively. \par
Then, we prove that the Wasserstein distance in \eqref{special_w_dis} is identical to the distance in \eqref{W-distance} when the decoder is the deterministic function $g$ in \eqref{residual}. We can rewrite the Wasserstein distance in \eqref{special_w_dis} with respect to the set of marginals of $\Gamma \in \mathcal{P}_{Y,\hat Y, W|X}$ on $Y$ and $\hat Y$ and denote this joint distribution set as $\mathcal{P}_{(Y,\hat Y)|X}$. It is evident that, given the marginal distributions of $Y$ and $\hat Y$, the set $\mathcal{P}_{(Y,\hat Y)|X} \subseteq \mathcal{P}_{Y,\hat Y|X}$, as the former set is subject to additional constraints on the joint distribution. On the other hand, from the topology of the graph, $\hat Y$ and $Y$ are conditionally independent given $W$ and $X$. Under a fixed mapping $g$ and a state $x$, the variable $W$ induces the marginal distribution of $\hat Y$. Consequently, the conditional distribution $P(W|Y)$ determines the coupling between $\hat Y$ and $Y$. Due to the generality of the encoder in representing any valid $P(W|Y)$, we can conclude that $\mathcal{P}_{Y,\hat Y|X} \subseteq \mathcal{P}_{(Y,\hat Y)|X}$, and furthermore $ W_1(P_q,P_g) =W^\ast_1(P_q,P_g)$. Next, we can derive the following result: 
\begin{equation*}
   \begin{aligned}
W_1& (P_q,P_g) = \inf_{\Gamma \in \mathcal{P}_{Y,\hat Y, W|X} }\mathbb{E}_{\Gamma} \left [ c(Y, \hat Y) \right ] \\
=& \inf_{\Gamma \in \mathcal{P}_{Y,\hat Y, W|X} } \iiint c(y, \hat y) p(y,\hat y, w |x)\mathrm {d}y \mathrm {d} \hat y \mathrm {d}w \\
=& \inf_{\Gamma \in \mathcal{P}_{Y,\hat Y, W|X} } \iiint c(y, \hat y) p(y|x)p(w|y)p(\hat y |w,x)\mathrm {d}y \mathrm {d} \hat y \mathrm {d}w \\ 
\end{aligned} 
\end{equation*}
\begin{equation}
   \begin{aligned}
= & \inf_{\Gamma \in \mathcal{P}_{Y, W|X} } \iint c(y, g(w,x)) p(y|x)p(w|y)\mathrm {d}y \mathrm {d}w \\  
= &  \inf_{ Q } \ \mathbb{E}_{P_q} \mathbb{E}_{Q(W|Y)} \left [ c\left (  Y,g(W,X)\right )  \right ] \\
& \ \ \ \ \ \ \mathrm{s.t.}  \ \ Q_w = P_w,
\end{aligned} 
\end{equation}
where $ \mathcal{P}_{Y, W|X}$ denotes the set of measures, subject to the constraint that their marginals correspond to the distributions of $Y$, and $W$ conditioned on $X$, respectively. In the final step, we only need to find a distribution $Q$ such that $\mathbb{E}_{P_y}\left [ Q(W | Y) \right ] = P_w$. It is important to note that $P_y$ represents the marginal distribution, while in the final equality, $P_g$ denotes the conditional probability of $Y$ given $X$. 
\end{proof}
In the equivalent distance function in \eqref{w-loss}, there exists an additional constraint on $Q_z$. In the numerical optimization during the training, we favor relaxing this constraint to a penalty term in the distance function, which we can express as:
\begin{equation} \label{w-loss-complete}
\inf_{ Q  } \ \mathbb{E}_{P_q} \mathbb{E}_{Q(W|Y)} \left [ c\left (  Y,g(W,X)\right )  \right ] + \lambda \cdot  D(Q_w,P_w),
\end{equation}
where $\lambda$ is a hyperparameter, and function $D$ serves as a discrepancy function that quantifies the divergence between the two distributions, yielding a value of zero when the distributions are identical. One can choose the discrepancy function as Jensen-Shannon (JS) divergence, approximated by adversarial learning to capture the distribution divergence \cite{goodfellow2014generative}. In this work, we employ another approach based on MMD \cite{dziugaite2015training}. Specifically, we assess the difference between the two distributions by comparing their mean embeddings in a reproducing kernel Hilbert space. Since MMD does not rely on specific distributional forms, one can employ it for two-sample testing, which offers advantages when training on given data samples. Moreover, as noted in \cite{tolstikhin2017wasserstein}, MMD yields more stable training performance compared to the JS divergence. We will discuss further details regarding the MMD method in subsequent sections. \par
Ultimately, our objective is to minimize the Wasserstein distance between the generated distribution and the true residual distribution, which we can express as:
\begin{equation} \label{objective_wasser}
    \inf_{ g, \ Q  } \ \mathbb{E}_{P_x}\mathbb{E}_{P_q} \mathbb{E}_{Q(W|Y)} \left [ c\left (  Y,g(W,X)\right )  \right ] + \lambda \cdot  D(Q_w,P_w),
\end{equation}
where $P_x$ denotes the empirical distribution of the state. Consequently, the training process draws data from the joint sample probability distribution over $X$ and $Y$. The training set consists of realizations from this distribution $\left \{ (x_1,y_1), \dots , (x_N,y_N) \right \} $. Furthermore, as indicated in \eqref{objective_wasser}, we approximate the Wasserstein distance between the two distributions by parameterizing the encoder $Q$ and the decoder $g$ to minimize the objective and thereby obtain a generated distribution that closely matches the true distribution. \par
In this framework, in addition to the basic autoencoder process, where $Q$ represents the encoder and $g$ the decoder, we also require that the samples of $w$ generated by the decoder conform as closely as possible to the target distribution $P_w$, i.e., a standard normal distribution. Specifically, in each iteration, we draw $n$ samples from the training set, with $\left \{ y_1, \dots , y_n \right \} $ approximately distributed according to the marginal distribution of $Y$. Concurrently, by passing these samples through the encoder, we obtain $n$ generated samples of $W$. Meanwhile, we draw $n$ samples from the standard normal distribution as the ground truth. Applying the MMD to these samples then approximates the difference between the distributions.
\begin{remark}
In this work, we adopt a common assumption that the prior distribution $P_w$ is standard normal and independent of the state $X$. This assumption not only simplifies the implementation but also improves training stability. Additionally, our objective is to model the effect of the introduced external disturbances on the system state, as characterized by the residual in \eqref{residual}. However, employing a conditional prior $ Pw(W|X)$ generally offers greater flexibility, as it allows for capturing the latent structure under varying conditions.
\end{remark}
\subsection{Uncertainty Propagation in Nonlinear Dynamics }
In SMPC problems, accurately characterizing uncertainty propagation through nonlinear systems over the prediction horizon remains a central challenge, especially when individual disturbances affect the dynamics at each step. PCE offers a structured framework to represent the uncertainty propagation using orthogonal polynomial bases. This approach facilitates the efficient approximation of state distributions and statistical moments. \par
In PCE, we consider a random variable $W$ with probability density function $p(w)$. Our focus lies within the weighted function space: 
\begin{equation} \label{funciton_space}
    \mathcal{F}:= \left \{ f:\mathbb{R}^{n_w}\to \mathbb{R}:\mathbb{E}_P\left [ f^2(w) \right ]< \infty    \right \} ,
\end{equation}
which is a separable Hilbert space equipped with the inner product $\left \langle f,g \right \rangle :=\mathbb{E}_P\left [  f(w)g(w)\right ]   $. By results from Hilbert space theory, there exists a complete orthonormal basis $\left \{ \phi_i \right \} _{i=0}^\infty \subset \mathcal{F} $, such that $\left \langle \phi_i, \phi_j \right \rangle = \delta_{ij} $ with $\delta_{ij}$ be the Kronecker delta. Then we can expand any square-integrable function $f \in \mathcal{F} $ uniquely in terms of the basis of the form: 
\begin{equation} \label{fcn_expansion}
f(w)=  {\textstyle \sum_{i=0}^{\infty}} \ d_i \phi_i(w),
\end{equation}
where projecting $f$ onto $\phi_i$ gives the coefficients $d_i= \left \langle  f,\phi_i\right \rangle $. In practice, the distribution of $W$ determines the choice of basis functions. For instance, when $W \sim \mathcal{N}(0,1) $, Hermite polynomials form the orthogonal basis, known as Wiener chaos \cite{luo2006wiener}. For multivariate random variables $w \in \mathbb{R}^{n_w}$, the multivariate orthonormal basis arises from the tensor product  
$\Phi_\alpha(w):= {\textstyle \prod_{i = 1}^{n_w}} \phi_{\alpha_i}(w_i), \quad \alpha \in \mathbb{Z}_{\ge 0}^{n_w}$. Furthermore, we generally need to truncate the expansion to a finite order $p$, and then approximate the function $f$ as $f(w)\approx  {\textstyle \sum_{\alpha\in\mathcal{A}_p }}  \ d_{\alpha}\Phi_\alpha(w) $, where $d_{\alpha}$ is projection coefficient on basis $\Phi_\alpha(w)$, and the set $\mathcal{A}_p :=\left \{ \alpha \in \mathbb{Z}_{\geq0}^{n_w}:  {\textstyle \sum_{i=1}^{n_w}}\alpha_i\leq p   \right \}$. The number of multivariate polynomial basis functions up to total degree $p$ in $n_w$ dimensions follows $\left | \mathcal{A}_p  \right | = \frac{(n_w+p)!}{n_w!p!} $. In the following, for the sake of formulation simplicity, we adopt the form $f(w)\approx  {\textstyle \sum_{i = 0 }^{P-1}}  \ \hat{w}_i\Phi_i(w)$ to represent the multivariate PCE with $P = \left | \mathcal{A}_p  \right |$, and $\hat{w}_i$ denotes the corresponding coefficient. \par
In addition, based on the orthogonality property of the PCE basis, we compute the mean and variance of the truncated approximation of the function $f$ as shown below:
\begin{equation} \label{PCE_mean}
    \mathbb{E}\left [ f(w) \right ]  = \hat{w}_0,
\end{equation}
\begin{equation} \label{PCE_variance}
    \mathrm{Var } \left [ f(w) \right ]  =  {\textstyle \sum_{i=1}^{P-1}} \hat{w}_i^2\left \langle \Phi_i(w), \Phi_i(w) \right \rangle.
\end{equation}
With normalization, the basis functions become orthonormal, having inner products equal to one.\par

Next, we discuss how to determine the PCE coefficients. Generally, we classify the computation of the PCE coefficients into two approaches, namely non-intrusive and intrusive methods \cite{eldred2009comparison}. The non-intrusive approach treats the system as a black-box and estimates the PCE coefficients based on realizations of the uncertain sample input. In contrast, the intrusive approach rewrites the transition model in terms of the PCE basis, yielding a coupled deterministic system that governs the evolution of the PCE coefficients. Although the intrusive method may offer better accuracy and preserve structure, it requires an analytical formulation of the system equations, which can be challenging for complex models. \par
In this work, we adopt the non-intrusive PCE approach, which offers greater flexibility and ease of integration in the SMPC framework. More precisely, we follow the approach outlined in \cite{fagiano2012nonlinear}, performing a weighted least squares regression on a finite set of collocation points, similar to the Probabilistic Collocation Method \cite{hockenberry2004evaluation}. First, $\bar{w} ^{(i)} \in \mathbb{R}^{n_w} $ with $i \in \left \{ 1:N_w \right \} $ represent $N_w$ independent samples drawn from its probability distribution. The total number of polynomial terms $P$ depends on the dimensionality $n_w$ of $w$ and the chosen maximal truncation degree $p$. Given the function $f: \mathbb{R}^{n_w}\rightarrow   \mathbb{R}$, we can evaluate the function $f$ at each sample $\bar{w} ^{(i)} $ and denote the corresponding outputs by $\bar{\nu}^{(i)}$.\par
In the following, we consider a least-squares regression problem with $l_2$-norm regularization of the form:
\begin{equation} \label{ls_problem}
    \min_{\hat{w}} \ \lambda\left \| \Gamma \left ( \bar{\nu}- \bar{\Phi}\hat{w}\right )  \right \| _2^2 + \left \| W\hat{w} \right \|_2^2
\end{equation}
where $\hat{w} \in \mathbb{R}^P $ denotes the vector of basis coefficients, $\bar{\nu} \in \mathbb{R}^{N_w} $ represents the simulated output sample vector, and $\bar{\Phi} \in \mathbb{R}^{N_w \times P}$ is the matrix containing the evaluations of the basis functions at different input samples. Moreover, $\Gamma \in \mathbb{R}^{N_w \times N_w } $ denotes a diagonal matrix whose diagonal entries correspond to the probability density values associated with each sample point, thereby assigning higher importance to samples with higher probability. To address the typically small magnitude of probability density function values, particularly in high-dimensional spaces, we introduce a scaling factor $\lambda$ to rescale the weights. Furthermore, $W \in \mathbb{R}^{P\times P} $ is a diagonal matrix, representing the regularization penalty to the basis coefficient. In this case, assigning larger weights to higher-order polynomials promotes smoothness and reduces overfitting. \par
Ultimately, the analytical solution to \eqref{ls_problem} is as follows:
\begin{equation} \label{ls_solution}
  \hat{w}^\ast = \underset{:= A}{\underbrace{\left ( \lambda \bar{\Phi}^\top \Gamma^\top\Gamma \bar{\Phi}  + W^\top W\right ) ^{-1} \lambda\bar{\Phi}^\top \Gamma^\top\Gamma } }  \ \bar{\nu}, 
\end{equation}
with $A \in \mathbb{R}^{P\times N_w}$. In cases where the number of basis functions exceeds the number of available samples, i.e., $P>N_w$, we can apply the Woodbury matrix identity to reformulate the inverse as $\left ( \lambda \bar{\Phi}^\top \Gamma^\top\Gamma \bar{\Phi}  + W^\top W\right ) ^{-1} = W^{-2}-\lambda W^{-2} \bar{\Phi}^\top \Gamma\left ( I+\lambda  \Gamma \bar{\Phi} W^{-2}\bar{\Phi}^\top\Gamma\right )^{-1} \Gamma \bar{\Phi} W^{-2}$, which reduces the computational burden from the direct inversion of a $P \times P$ matrix. \par
In this recursive PCE approach, although the state $x_{i+1|k}$ depends on the entire history of past disturbances, we approximate it as a function of only the current random input $w_{i|k}$ and the current state $x_{i|k}$. For control purposes, it is desirable to design a control strategy that progressively contracts the state distribution at each step. Therefore, in this work, we decouple the dependence on the full disturbance history and focus only on the overall dispersion and shape characteristics of the marginal state distribution of $x_{i|k}$.  \par
In detail, we consider the propagation of uncertainty over the prediction horizon at time instant $k$. At time $0|k$, the state $x_{0|k}$ is then equal to the current state feedback $x_k$, and the uncertainty arises from $w_{0|k}$, which is a latent variable assumed to follow a Gaussian distribution. Sampling a total of $N_{w_0}$ samples $w_{0|k}^{(j)}\in \mathbb{R}^{n_w} $ from this distribution, we simulate the model to obtain the corresponding next-step states $x_{1|k}^{(j)} \in \mathbb{R}^{n_x} $. Additionally, we denote $\bar{x}_{1|k} \in \mathbb{R}^{N_{w_0} \times n_{x}}  $ as the simulated output sample matrix. The fitted coefficients $\hat{x}_{1|k}\in \mathbb{R}^{P\times n_x} $ based on \eqref{ls_solution} are:
\begin{equation} \label{ls_solution_w0}
\hat{x}_{1|k} = A \bar{x}_{1|k} .
\end{equation} \par
At time step $i|k$ with $i>0$, the uncertainty arises not only from $w_{i|k}$, but also from all previous disturbances $w_{i-1|k}, \dots , w_{0|k}$. Accordingly, when constructing the PCE for $x_{i+1|k}$, the polynomial basis must depend on all $w_{\kappa | k}$ for $\kappa = 0, \dots, i$. As a result, if the polynomial order $p $ remains fixed, the total number of polynomial terms increases rapidly with the number of time steps. Specifically, for small $p$, the total number of basis terms increases polynomially as $P(i) = \binom{p+(i+1)n_w }{(i+1)n_w }  \sim \left (  (i+1)n_w\right ) ^ p$. Moreover, the number of required samples increases exponentially with the prediction horizon, rendering the PCE construction for $x_{i+1|k}$ computationally intractable. \par
In the proposed PCE-based method, we decouple the influence of historical  disturbances by incorporating it into the marginal distribution of the state $x_{i|k}$. As a result, $x_{i+1|k}$  depends only on $x_{i|k}$ and the current disturbance $w_{i|k}$. Specifically, in the optimization problem defined in \eqref{smpc}, we consider only the first and second moments of the marginal state distribution $x_{i|k}$ and treat them as implicit decision variables. Here, $\underline{x} _{i|k} \in \mathbb{R}^{n_x} $ denotes the first moment, i.e., the mean vector, and $\Sigma _{i|k} \in \mathbb{R}^{n_x \times n_x} $ denotes the second moment, i.e., the covariance matrix, at the step $i$. \par
In practice, considering the symmetry of the covariance matrix $\Sigma _{i|k}$ and the reparameterization trick used later, we adopt the lower-triangular matrix $L _{i|k} \in \mathbb{R}^{n_x \times n_x}$ from the Cholesky decomposition of $\Sigma _{i|k}$  as the decision variable, where $\Sigma _{i|k} = L _{i|k} L  _{i|k} ^ \top$. This approach reduces the number of variables to $n_x(n_x+1)/2$, guarantees symmetry and positive definiteness of the covariance, and improves numerical stability. To further ensure the uniqueness of the decomposition, we impose non-negativity constraints on the diagonal elements of $L _{i|k} $. Under this assumption, the maximum entropy distribution estimate given the specified first and second moments is a normal distribution. Thus, we introduce a new latent random variable $\xi$ following a standard normal distribution. To enable differentiable sampling, we use the reparameterization and represent state samples as: 
\begin{equation} \label{state_samples}
    x _{i|k}^{(j)}  = L_{i|k} \ \xi_{i|k}^{(j)} + \underline{x} _{i|k}.
\end{equation}
According to \eqref{ls_solution_w0}, we can similarly derive the PCE coefficients $\hat{x}_{i|k}$. Furthermore, we can partition these state coefficients into two parts. The first part, $\tilde {x}_{i|k} := \hat{x}_{i|k}[0,:]^{\top } \in \mathbb{R}^{n_x} $, represents the mean value of $x_{i|k}$ based on \eqref{PCE_mean}. The corresponding nonlinear equality constraint takes the form:
\begin{equation} \label{mean_constraint}
  \tilde {x}_{i|k} - \underline{x} _{i|k} = 0.
\end{equation}
The second part is $\tilde {X}_{i|k} :=  \hat{x}_{i|k}[1:P,:]\in \mathbb{R}^{P-1 \times n_x}$, and $\tilde {X}_{i|k} ^\top  \tilde {X}_{i|k}$ captures the covariance of $x_{i|k}$ based on \eqref{PCE_variance}. The corresponding nonlinear equality constraint has the form:
\begin{equation} \label{covariance_constraint}
 \tilde {X}_{i|k} ^\top  \tilde {X}_{i|k} - L_{i|k}L_{i|k}^{\top } = 0.
\end{equation}
This defines a symmetric matrix equality constraint, which we equivalently decompose into $n_x(n_x + 1)/2$ scalar equality constraints. In practical implementations, we can also relax this constraint and incorporate it into the objective function as an augmented term. Specifically, the relaxation introduces a Frobenius norm $\left \| \tilde {X}_{i|k} ^\top  \tilde {X}_{i|k} - L_{i|k}L_{i|k}^{\top }  \right \| _F^2$. \par
Subsequently, we consider the following time-varying control policy in this work:
\begin{equation} \label{control_policy}
u_{i|k} = \underline{u}_{i|k}+K_{i|k}(x_{i|k}-\underline{x}_{i|k}) 
\end{equation}
where both $\underline{u}_{i|k} \in \mathbb{R}^{n_u} $ and $K_{i|k} \in \mathbb{R } ^{n_u \times n_x }$ are decision variables. The matrix $K_{i|k}$ as feedback gain regulates the contraction of uncertainty propagation, while $\underline{u}_{i|k}$ serves as a feedforward term. For simplicity, we redefine the decision variables as $\tilde{K}_{i|k}$ and $\tilde{u}_{i|k}$, as follows. When combining \eqref{state_samples} and \eqref{control_policy}, we observe the emergence of a bilinear term $K_{i|k}L_{i|k}$. To address this, we redefine a feedback gain $\tilde{K}_{i|k}$ as: 
\begin{equation} \label{control_policy_redef}
\tilde K _{i|k} := K_{i|k}L_{i|k}.
\end{equation}
Meanwhile, we introduce a new variable $\tilde{u}_{i|k} :=  \underline{u}_{i|k} - K_{i|k}\underline{x}_{i|k}$ as the feedforward term owing to the arbitrariness of $ \underline{u}_{i|k}$. At this stage, we first sample the disturbance at each prediction step, generating $N_w$ realizations of the lifted uncertainty $\{  (\xi_{i|k}^{(j)},w_{i|k}^{(j)}) \}$. Accordingly, the propagation results of $x_{i+1|k}^{(j)}$ based on \eqref{nonlinear_ss} yield:
\begin{equation} \label{propapgation_state}
\begin{aligned}
    x_{i+1|k}^{(j)} = & f(L_{i|k}  \xi_{i|k}^{(j)} + \underline{x}_{i|k}, \tilde {K}_{i|k} \xi_{i|k}^{(j)} + \tilde {u}_{i|k} ) \\
    &+ g(L_{i|k}  \xi_{i|k}^{(j)} + \underline{x}_{i|k}, w_{i|k}^{(j)}),
\end{aligned}
\end{equation}
where the state and control input samples are affine functions of the decision variables, and the simulated output sample matrix as in \eqref{ls_solution_w0} is $\bar{x}_{i+1|k}=\left [ x_{i+1|k}^{(1) },  \dots  ,  x_{i+1|k}^{(N_w) } \right ] ^\top $. \par
Subsequently, we address the soft constraints to reformulate them into a tractable form.  As discussed above, we constrain only the first and second moments of the state distributions, as shown in \eqref{mean_constraint} and \eqref{covariance_constraint}, while potential deviations in higher-order moments are not explicitly considered. To account for this, we consider the family of distributions that share the same mean and covariance and assess the impact of the worst-case distribution within this family on the soft constraints, corresponding to a worst-case distributionally robust chance constraint under known mean and covariance. In this context, applying Cantelli’s inequality yields tractable conditions \cite{calafiore2006distributionally}, under which we reformulate the state and input constraints in \eqref{smpc} as:
\begin{equation} \label{state_constraint}
G_{x,i}\underline{x}_{i|k} + \kappa_{\epsilon_1}  \left \| L_{i|k} G_{x,i}^\top  \right \| _2 \leq g_{x,i}, \ \text{with} \ \kappa_{\epsilon_1 } = \sqrt{\frac{1-\epsilon_1}{\epsilon_1} }  ,
\end{equation}
and
\begin{equation} \label{input_constraint}
G_{u,i}\tilde{u}_{i|k} + \kappa_{\epsilon_2}  \left \| \tilde {K}_{i|k} G_{u,i}^\top  \right \| _2 \leq g_{u,i}, \ \text{with} \ \kappa_{\epsilon_2 } = \sqrt{\frac{1-\epsilon_2}{\epsilon_2} }  ,
\end{equation}
where the reformulated constraints are all second-order cone constraints. Here, $G_{x,i}$ denotes the $i$-th row of the matrix $G_x$, indicating that each constraint is enforced individually for each component rather than jointly enforcing a probabilistic guarantee over the whole set $\mathcal{X}  $. Moreover, we observe from the formulation in \eqref{input_constraint} that the redefinition of the decision variables in \eqref{control_policy_redef} simplifies the overall expression.

\subsection{Objective Function}
In this section, we introduce a new variant of the objective function to incorporate the higher-order moments of the variables. In standard SMPC frameworks, the objective aims to optimize the performance in the expected sense. When adopting a quadratic form for the stage cost, the expression for the performance index in \eqref{smpc} follows as:
\begin{equation} \label{objective_standard}
J_k := \mathbb{E}\left [ \sum_{i=0}^{N-1} \left ( x_{i|k}^\top Q x_{i|k} + u_{i|k}^\top R u_{i|k}\right )+x_{N|k}^\top Q_T x_{N|k}  \right ]  ,
\end{equation}
where $Q\succ 0$, $R\succ 0$, and $Q_T\succ 0$ denote the weighting matrices for the state, input, and terminal state, respectively. The terminal cost matrix $Q_T$ arises from the solution to the discrete-time Lyapunov equation $Q_T =Q + K^\top R K +(A+BK)^\top Q_T (A+BK)$, where we obtain the feedback gain matrix $K$ by solving the discrete-time algebraic Riccati equation, given the weighting matrices $Q$, $R$, and the linearized system around the mean operating point. Moreover, given the mean and covariance of the state and input, the expected value $\mathbb{E}\left [ x_{i|k}^\top Q x_{i|k} \right ]  $ takes the form $\mathbb{E}\left [ x_{i|k}^\top Q x_{i|k} \right ]  = \mathbb{E}\left [  x_{i|k}\right ]  ^\top Q \mathbb{E}\left [  x_{i|k}\right ]+\mathrm {Tr}\left ( Q \mathrm {Cov}(x_{i|k}) \right ) $. Then, the alternative formulation of the objective function in \eqref{objective_standard} takes the form: 
\begin{equation} \label{objective_ref}
\begin{aligned}
     J_k & =\sum_{i=0}^{N-1}\Big(  \mathbb{E}\left [  x_{i|k}\right ]  ^\top Q \ \mathbb{E}\left [  x_{i|k}\right ]+\mathrm {Tr} ( Q  \ \mathrm {Cov}(x_{i|k})    )\\
     & +  \mathbb{E}\left [  u_{i|k}\right ]  ^\top R \ \mathbb{E}\left [  u_{i|k}\right ] +\mathrm {Tr} ( R\ \mathrm {Cov}(u_{i|k})    )  \Big)\\
     & + \mathbb{E}\left [  x_{N|k}\right ]  ^\top Q_T \ \mathbb{E}\left [  x_{N|k}\right ]+\mathrm {Tr} ( Q_T \ \mathrm {Cov}(x_{N|k})    )  .
\end{aligned}
\end{equation} 
By examining the formulation in \eqref{objective_ref}, we observe that only the first and second moments of the state and input distributions are penalized at each prediction step, without explicitly considering higher-order moments. \par
In the following, we introduce a novel objective function that penalizes the distance between the actual distribution and a desired target distribution. To this end, we first present the definition of MMD, given as follows:
\begin{equation} \label{MMD}
\mathrm {MMD}_{\mathcal{H} }(P,Q)=\sup_{f\in\mathcal{H}, \left \| f \right \|_{\mathcal{H}} \le 1 } \left ( \mathbb{E}_{x\sim P}\left [ f(x) \right ]-  \mathbb{E}_{y\sim Q}\left [ f(y) \right ] \right )   ,
\end{equation}
where $ \mathcal{H}$ denotes a reproducing kernel Hilbert space (RKHS), and $f$ belongs to this space. Intuitively, we can interpret MMD as finding the function $f$ that best distinguishes between the distributions $P$ and $Q$, thereby quantifying their difference. We restrict the search to the family of functions within the RKHS whose norms do not exceed one. Furthermore, When $f$ is linear, it measures discrepancies in the first-order moments, and when $f$ is quadratic, it captures second-order moment differences. By considering a sufficiently rich function family, MMD provides a measure of higher-order discrepancies between the distributions. \par
Subsequently, by leveraging the reproducing property of the RKHS and the Cauchy–Schwarz inequality, we can derive the following expression for the MMD:
\begin{equation} \label{MMD_ref}
\mathrm {MMD}_{\mathcal{H} }(P,Q)=\left \| \mu _P-\mu_Q \right \|_{\mathcal{H} } ,
\end{equation}
where $\mu_P:=\mathbb{E}_P \left [ k(x,\cdot ) \right ] $ defines the mean embedding of the distribution $P$, and similarly for $Q$. The bivariate function $k(\cdot, \cdot)$ serves as the reproducing kernel that determines the geometry of the RKHS. The mean embedding captures the mapping of the probability density function of $P$ and $Q$ from the $\mathcal{L}^2 $ space into the RKHS. Under this framework, the MMD provides an effective measure of the distance between the distributions $P$ and $Q$ in the RKHS. In practice, we often use the squared MMD expanded as:
\begin{equation} \label{MMD_2}
\begin{aligned}
     \mathrm {MMD}_{\mathcal{H} }^2 &= \mathbb{E}_{x,x' \sim P}\left [ k(x,x') \right ]  + \mathbb{E}_{y,y' \sim Q}\left [ k(y,y') \right ]\\
     &-2\mathbb{E}_{x \sim P,y \sim Q}\left [ k(x,y) \right ].
\end{aligned}
\end{equation} \par
Furthermore, we can obtain a finite sample estimate of \eqref{MMD_2} using empirical samples. Given samples $\{x^{(i)}\} _{i=1} ^{N_x}$ from $P$ and $\{y^{(j)}\} _{j=1} ^{N_y}$ from distribution $Q$, the empirical estimate of the MMD takes the form:
\begin{equation} \label{MMD_final}
\begin{aligned}
     \mathrm {MMD}_{\mathcal{H} }^2 &= \frac{1}{N_x(N_x-1)}\sum _{i\ne j} k(x^{(i)},x^{(j)})  + \frac{1}{N_y(N_y-1)}\\&\sum _{i\ne j} k(y^{(i)},y^{(j)})
     -\frac{1}{N_xN_y}\sum _{i, j} k(x^{(i)},y^{(j)}).
\end{aligned}
\end{equation} \par
Next, we utilize the MMD to formulate our new objective function. Our goal is to drive the distribution of the state $x_{i|k}$ at each prediction step to align closely with a prescribed target distribution, represented by a Dirac delta distribution $\delta(x-x_{\mathrm{ref}})$ with the reference value $x_{\mathrm{ref}}$. Accordingly, the deviation between the state distribution and the target based on \eqref{MMD_2} takes the form: 
\begin{equation} \label{MMD_target}
\begin{aligned}
     \mathrm {MMD}^2(\delta_{x_{\mathrm{ref}}},P(x_{i|k}))&=k(x_{\mathrm{ref}},x_{\mathrm{ref}}) + \mathbb{E}_{x,x' \sim P}\left [ k(x,x') \right ] \\& -2 \mathbb{E}_{x\sim P }\left [ k(x_{\mathrm{ref}},x) \right ] 
\end{aligned}
\end{equation}
where the first term constitutes a constant, which we omit in the optimization. Then, we numerically estimate the squared MMD in \eqref{MMD_target} using the samples $\bar{x}_{i|k}$ generated from the simulations as in \eqref{ls_solution_w0}. \par
Since the estimation method based on \eqref{MMD_final} is not data-efficient, we further employ a PCE-based expansion of the MMD to enable fast computation while avoiding extensive sampling. Specifically, we precompute the least squares projection matrix $A$ as in \eqref{ls_solution_w0} offline, following the above-mentioned procedure. Then, by evaluating the MMD for each sample in $\bar{x}_{i|k}$, we obtain the simulated output, denoted as $\bar{\chi}_{i|k} \in \mathbb{R}^{N_w} $. The PCE coefficients are $\hat{\chi}_{i|k} \in \mathbb{R}^{P}  =A\bar{\chi}_{i|k}$ through \eqref{ls_solution_w0}. In practice, our interest lies in the expected value of the kernel function in \eqref{MMD_target}, which corresponds to the first PCE coefficient. Therefore, it is sufficient to retain only the truncated matrix $A_{1,:}$ for computation. \par
The subsequent question is how to evaluate the MMD contribution for each sample efficiently. Here, we treat each sample $x_{i|k}^{(j)}$ in $\bar{x}_{i|k}$ as a singleton distribution and approximate its contribution to the local MMD. Then, we can express the second term as $\frac{1}{N_w-1} \sum _{j \ne z} k(x_{i|k}^{(j)}, x_{i|k}^{(z)})$ and the third term as $k(x_{\mathrm{ref} },x_{i|k}^{(j)})$. This yields the MMD evaluation for individual samples, which we then aggregate to form $\bar{\chi}_{i|k}$. In addition, to avoid redundant evaluations of the kernel function $k$, we store the computed values in a kernel matrix defined as $K_{\mathrm{MMD}} = \left \{ k(x_i,x_j) \right \}_{i,j=1} ^{N_w}$. \par
As a result, the objective function takes the form:
\begin{equation} \label{MMD_obj}
\begin{aligned}
J_{ \mathrm{MMD}  } = \sum_{i=1}^{N} \mathrm {MMD}^2(\delta_{x_{\mathrm{ref} }},P(x_{i|k})) + \left \|  \mathbb{E}\left [ u(i|k) \right ]  \right \| _{R} ^2.
\end{aligned}
\end{equation}
By optimizing this objective function, we can control the distribution of the state at each prediction step to contract towards the reference $x_{\mathrm{ref} }$.

\section{Simulation Results and Discussion}
In this section, we begin by introducing the system states, control inputs, boundary conditions, and the control objectives. We then detail the setup, including parameter choices and implementation settings. Next, we validate the learned uncertainty model in terms of distributional alignment and assess the closed-loop control performance of various MPC strategies under consistent test conditions. Finally, we provide a discussion of the results, highlighting key insights and limitations.
\subsection{System Description}
In this work, we employ a data-driven model from \cite{chen2023gasoline} as both the prediction model and the plant model for control, primarily to evaluate the performance of the proposed control strategy. The model describes the state evolution between engine cycles. The study collects all experimental data from a single-cylinder research engine designed for HCCI combustion research \cite{bedei2023dynamic}. \par
The system states include three key performance indicators of HCCI combustion: $x_{\mathrm{CA50}}$, $x_{\mathrm{IMEP}}$, and $x_{\mathrm{DPmax}}$. Specifically, CA50 represents the crank angle at which \SI{50}{\percent} of the total heat release has occurred. It serves as an indicator of combustion efficiency, where a large CA50 value suggests inefficient combustion, and a small or negative value indicates early combustion. Both excessively delayed and early combustion are undesirable. Moreover, the variation in CA50 reflects combustion stability, with smaller variations indicating more consistent combustion at the operating point.  Indicated Mean Effective Pressure (IMEP), represents the average pressure in the cylinder and serves as a direct measure of engine output torque. DPmax refers to the maximum pressure rise rate during combustion. High DPmax values, often associated with phenomena such as knocking, can result in excessive mechanical stress and noise. The in-cylinder pressure trace provides all the information needed to calculate these state variables.\par
The corresponding control objectives and constraints are defined as follows: 
For CA50, we aim to constrain the state within the safe interval 
$x_{\mathrm{CA50}} \in \left [ \SI{2}{\degreeCA}    , \SI{13}{\degreeCA} \right ] $, while maintaining it near a setpoint of \SI{7}{\degreeCA}. For IMEP, the control objective is accurate reference tracking, aiming to closely follow a specified torque demand. We constrain the IMEP state within the operating range $x_{\mathrm{IMEP}} \in \left [ \SI{2}{\bar}   , \SI{4.5}{\bar} \right ] $. For DPmax, we enforce a safety constraint to avoid excessive pressure rise, with a range $x_{\mathrm{DPmax}} \in \left [ \SI{0}{\bar\per\degreeCA}    , \SI{5}{\bar\per\degreeCA} \right ] $. \par

Regarding the control inputs, the system has three actuators, namely $u_{\mathrm{nvo}}$, $u_{\mathrm{fuel}}$, and $u_{\mathrm{eth}}$. $u_{\mathrm{nvo}}$ denotes the Negative Valve Overlap (NVO), which represents the timing advance of the exhaust valve closing relative to the intake valve opening. This creates a period where residual gases from the previous cycle remain trapped inside the cylinder, which is essential for enabling stable HCCI combustion. By adjusting the NVO, we can modulate the amount of trapped residual gases, which in turn influences the combustion characteristics. We constrain the NVO input within $u_{\mathrm{nvo}} \in \left [ \SI{172}{\degreeCA}    , \SI{232}{\degreeCA} \right ] $. $u_{\mathrm{fuel}}$ represents the fuel injection duration, measured in milliseconds, which is approximately proportional to the injected fuel mass. In this experiment, the start of injection is fixed at \SI{456}{\degreeCA}, and the fuel injection input is constrained within $u_{\mathrm{fuel}} \in \left [ \SI{0.5}{\milli\second}  , \SI{0.97}{\milli\second} \right ] $. $u_{\mathrm{eth}}$ represents the amount of ethanol injection, also measured in milliseconds, with the fixed injection timing at \SI{580}{\degreeCA}. Ethanol injection serves not only as a supplementary fuel but also as a means to provide cooling effects. The ethanol injection input lies within $u_{\mathrm{eth}} \in \left [ \SI{0}{\milli\second}  , \SI{0.4}{\milli\second} \right ] $.\par 
The system dynamics result from the interactions between these state variables and control inputs. Moreover, the modeling approach treats unmodeled dynamics, external disturbances, and measurement noise collectively as random variables. The characterization of these uncertainties relies on residual data, computed as the discrepancy between the model predictions and experimental measurements. \par

\subsection{Setup}
This section outlines the setup, including the training procedure of the uncertainty model, the configuration of the controller, and the specification of the optimization parameters. \par
In this work, we primarily focus on modeling the influence of disturbances on the two key variables, CA50 and IMEP. We further assume that the distribution of the residuals depends on the current states of CA50 and IMEP. To model this relationship, we collect $50000$ pairs of state and residual data 
$(x^{(i)},y^{(i)})$, and define the residual as $y := g(x, w)$ as in \eqref{nonlinear_ss}. The objective is to explicitly learn an mapping of $g$ through training. \par 
To facilitate this, we introduce a latent variable $w \in \mathbb{R}^2 $ that follows a standard normal distribution. Furthermore, to train the model, we adopt a neural network architecture comprising an encoder and a decoder. The encoder is a fully connected neural network, with ReLU activation functions following each layer. It consists of three hidden layers, each with 30 neurons. A sampling layer is appended at the end, which outputs samples from a learned probability distribution based on the predicted mean and log-variance. This design enables stochastic sampling while maintaining differentiability for backpropagation, following the approach in \cite{kusner2017grammar}. The decoder is a simple parameterized deterministic function approximator, also implemented as a fully connected neural network with three hidden layers of 30 neurons each. For the training process, we use a mini-batch size of 320 and train for 15 epochs. The loss function consists of the reconstruction and the regularization term in \eqref{objective_wasser}. The reconstruction loss takes the form of the mean squared error between the predicted and true residuals. For the regularization, we employ an MMD term using a Gaussian kernel with a standard deviation of $0.5$ and a weighting coefficient $\lambda$ of $2.5$. Compared to other generative models such as GAN, the Wasserstein Autoencoder framework generally provides more stable training and yields samples that more closely match the true distribution \cite{tolstikhin2017wasserstein}. \par
In the SMPC optimization problem, we normalize all state variables and control inputs. The prediction horizon is set to $N = 4$, considering the fast cycle-to-cycle dynamics of the HCCI combustion process. The allowable violation margins $\epsilon_1$ and $\epsilon_2$ for both the state and input constraints are set to $0.95$. Since the optimal control input at the initial prediction step is directly applied to the system without a feedback term, we impose a hard input constraint at the initial step to account for the practical input saturation limits. To guarantee the feasibility of the optimization problem, we add global slack variables to softly relax the constraints. \par
In the objective function defined in \eqref{objective_ref}, the state weighting matrix is $Q = \mathrm {diag}(10,10,0.1)  $, the input weighting matrix is $R = \mathrm {diag}(0.2,1,0.5)  $, and we compute the terminal cost matrix $Q_T$ offline by solving a discrete-time Lyapunov equation based on a linearized model around the mean state and input. The resulting $Q_T$ matrix has diagonal elements $[10.32$, $11.4$, $0.18]$ and off-diagonal elements: $Q_T(1,2)=0.65$, $Q_T(1,3)=-0.04$, and $Q_T(2,3)=-0.15$. In the MMD-based objective function, the input weight $R$ remains $\mathrm {diag}(0.2,1,0.5) $. The MMD computation uses a Gaussian kernel with a bandwidth of $2.5$ to represent a family of smooth functions. \par
Next, we specify the PCE settings. At the initial prediction step $0|k$, we compute the least squares projection matrix $A_0$ in \eqref{ls_solution_w0} based on the samples. Specifically, we draw $20$ samples from the standard normal distribution of $w$, and a maximal polynomial order is set to $3$, resulting in $P_0=10$ polynomial terms. The scaling weight $\lambda_0$ in \eqref{ls_problem} is set to $200$. In the matrix $W$ in \eqref{ls_problem}, the weights assigned to the $0$th to $3$rd order polynomials are $1$, $0.3$, $0.1$, and $0.03$, respectively. For subsequent prediction steps, i.e., steps $i > 0$, we draw $45$ samples from the joint standard normal distribution of $(\xi,w)$ and employ a maximal polynomial order of $2$, resulting in $P=21$ polynomial terms. The scaling weight $\lambda$ is set to $2000$. In this case, the weights for the $0$th to $2$nd-order polynomials are set as $1$, $0.3$, and $0.1$, respectively. We generate all sampling points for PCE construction offline and precompute the least-squares projection matrices $A$ at each prediction step accordingly. During online execution, only function evaluations over the sampled inputs are necessary. \par
\subsection{Evaluation Protocol}
In the subsequent evaluation, we first assess the performance of the uncertainty model trained using the autoencoder framework, specifically examining its discrepancy from the target distribution. Next, we evaluate the impact of incorporating the uncertainty model and the PCE method into the controller design. Finally, we examine the overall performance improvement achieved by integrating the MMD-based objective into the MPC framework. \par
For the uncertainty model evaluation, we build a dataset consisting of $50000$ data pairs $(x^{(i)},r^{(i)})$ collected from our test bench, with $10000$ pairs used for testing the performance of the trained model and the remaining $40000$ for training. These $10000$ test pairs reflect the marginal distribution of the uncertainty, as the states are sampled under various diverse operating conditions, consistent with the overall data distribution. Since the uncertainty model characterizes the conditional probability distribution of the uncertainty given different state conditions, we further evaluate its performance by testing the model outputs under both typical single-cylinder operating conditions and more extreme scenarios and examine whether the predicted distributions of the model accurately reflect the behavior of the true underlying uncertainty. \par
To evaluate the effectiveness of the PCE-based method and the MMD-based objective function in the controller design, we conduct a series of ablation studies. Specifically, we compare the following control strategies. Nonlinear MPC: A baseline nonlinear MPC that does not explicitly account for uncertainty in the prediction; Gaussian-based SMPC: A stochastic MPC that models uncertainty using a Gaussian approximation, where the uncertainty model is linear and equipped with a quadratic objective; PC-based SMPC: A stochastic MPC that utilizes a nonlinear uncertainty model via PCE with a quadratic objective; GEM-SMPC: A stochastic MPC that incorporates both a nonlinear uncertainty model via PCE and a non-quadratic MMD-based objective. In this ablation study, we exclude the combination of SMPC with Gaussian approximation and MMD objectives, as this SMPC approach is not sample-based and the Gaussian distribution provides limited modeling flexibility. \par
In the evaluation of the controllers, we primarily focus on two key performance indicators: CA50 and IMEP. For CA50, we aim to ensure stable combustion by minimizing its variance, while also avoiding frequent violations of state constraints. For IMEP, our goal is accurate reference tracking. Because IMEP is less sensitive to uncertainty compared to CA50, we consider the mean IMEP trajectories in this work. In the simulation study, we simulate $120$ consecutive engine cycles with a time-varying IMEP reference. For each controller, we conduct $50$ Monte Carlo simulations, each with different disturbance realizations, to capture the variability of the control performance. In the setup, we fix the setpoint for CA50 at \SI{7}{\degreeCA} and design the IMEP reference as a step signal, switching between the values $2.8$, $2.2$, $3.2$, and \SI{3.9}{\bar}.

\subsection{Results}
This subsection presents the results of the evaluation, covering the sampling performance of the learning-based uncertainty model and the comparative performance of various controllers with respect to CA50 stability and IMEP tracking accuracy.  \par
First, we present the simulation results of the learned uncertainty model. Figure~\ref{exp:uncertainty} and Figure~\ref{exp:uncertainty2} illustrate the consistency between the ground truth and the distribution generated by the proposed uncertainty model in terms of CA50 and IMEP. The ground truth distribution arises from uncertainty samples collected under varying state operating conditions, while the generated samples are also produced under the identical set of state conditions. We use a total of $10000$ samples in each case, providing a reasonable approximation of the marginal distribution of uncertainty. \par
As shown in the scatter plot of Figure~\ref{exp:uncertainty}, the uncertainty model captures the general trend and spread of the ground truth distribution. 

\begin{figure}[htb]
\centering
\includegraphics[width=0.48\textwidth]{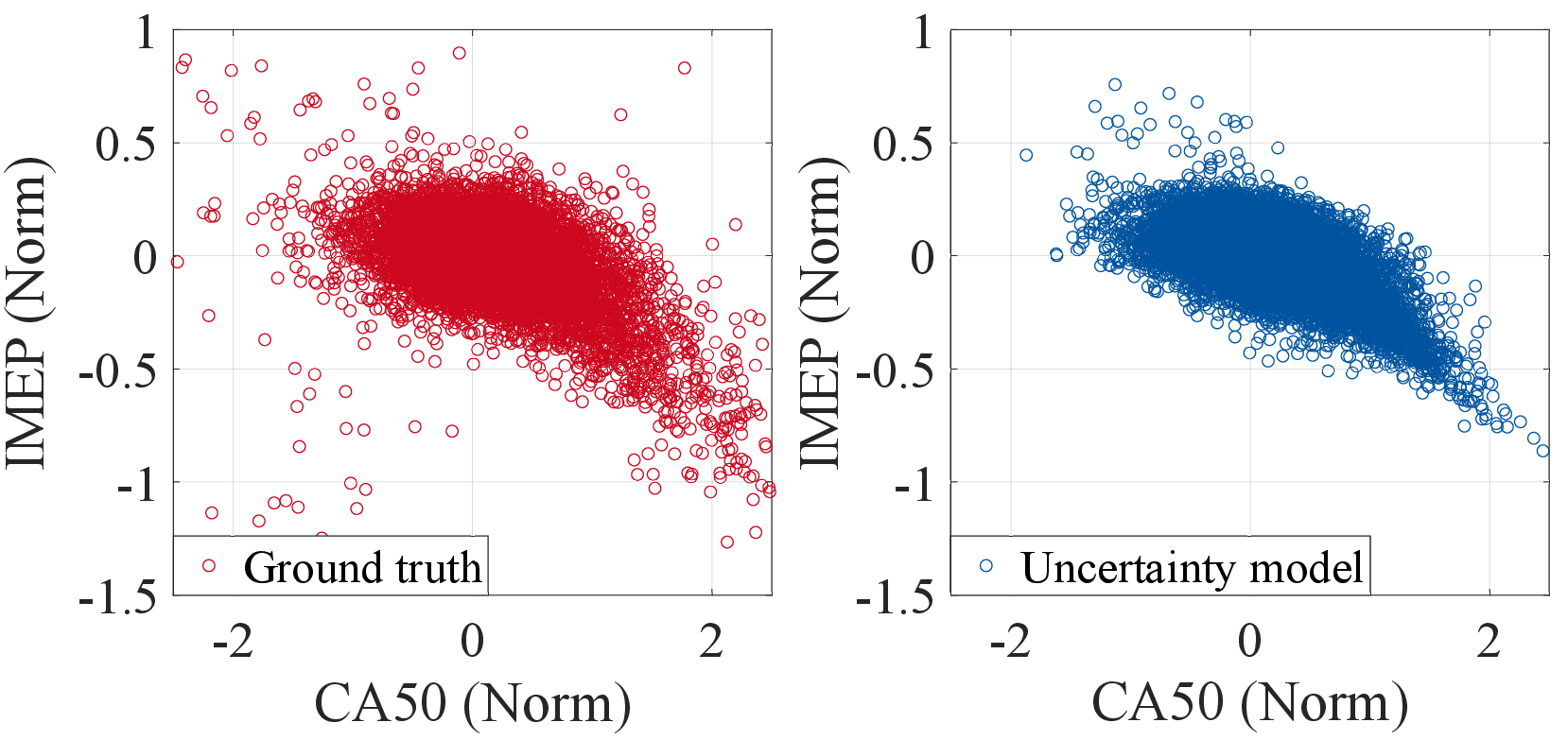}   
\caption{Comparison of the marginal uncertainty distributions of normalized IMEP and CA50 between the ground truth (left) and the uncertainty model (right), based on $10,000$ samples.  }  
\label{exp:uncertainty}  
\end{figure}      
In Figure~\ref{exp:uncertainty2}, the 3D histogram visualization further confirms the alignment of the peak and concentration areas between the modeled and true distributions. The generated distribution also reflects certain higher-order moment characteristics such as skewness and kurtosis. \par

\begin{figure}[htb]
\centering
\includegraphics[width=0.48\textwidth]{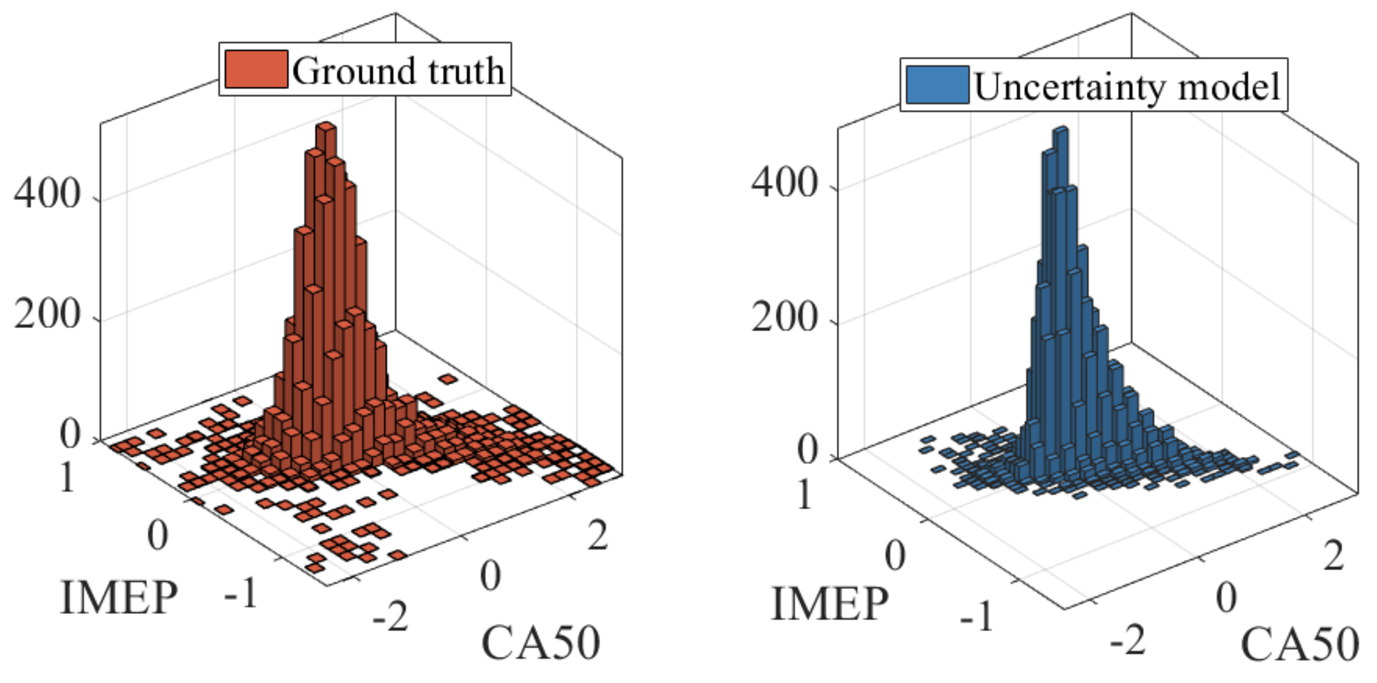}   
\caption{Comparison of the marginal uncertainty  distributions of normalized IMEP and CA50 between the ground truth (left) and the uncertainty model (right), visualized as 3D histograms based on $10,000$ samples.   }  
\label{exp:uncertainty2}  
\end{figure}   

Figure~\ref{exp:uncertainty_condition} illustrates the conditional uncertainty distributions of normalized CA50 and IMEP under two operating conditions. The upper plots correspond to a stable operating regime of the research cylinder with CA50 $= \SI{7}{\degreeCA}$ and IMEP$= \SI{3}{\bar}$, while the lower plot represents a boundary condition with CA50 $= \SI{2}{\degreeCA}$, IMEP $= \SI{2.2}{\bar}$, where combustion becomes more sensitive and less stable. In both cases, $10000$ samples were generated from the uncertainty model, conditioned on the respective operating points. In the stable regime, the distribution is compact and almost symmetric.
Under the boundary condition, the generated samples still capture the core trend of the empirical distribution, though some deviation and misalignment occur in the extremes.\par

\begin{figure}[htb]
\centering
\includegraphics[width=0.48\textwidth]{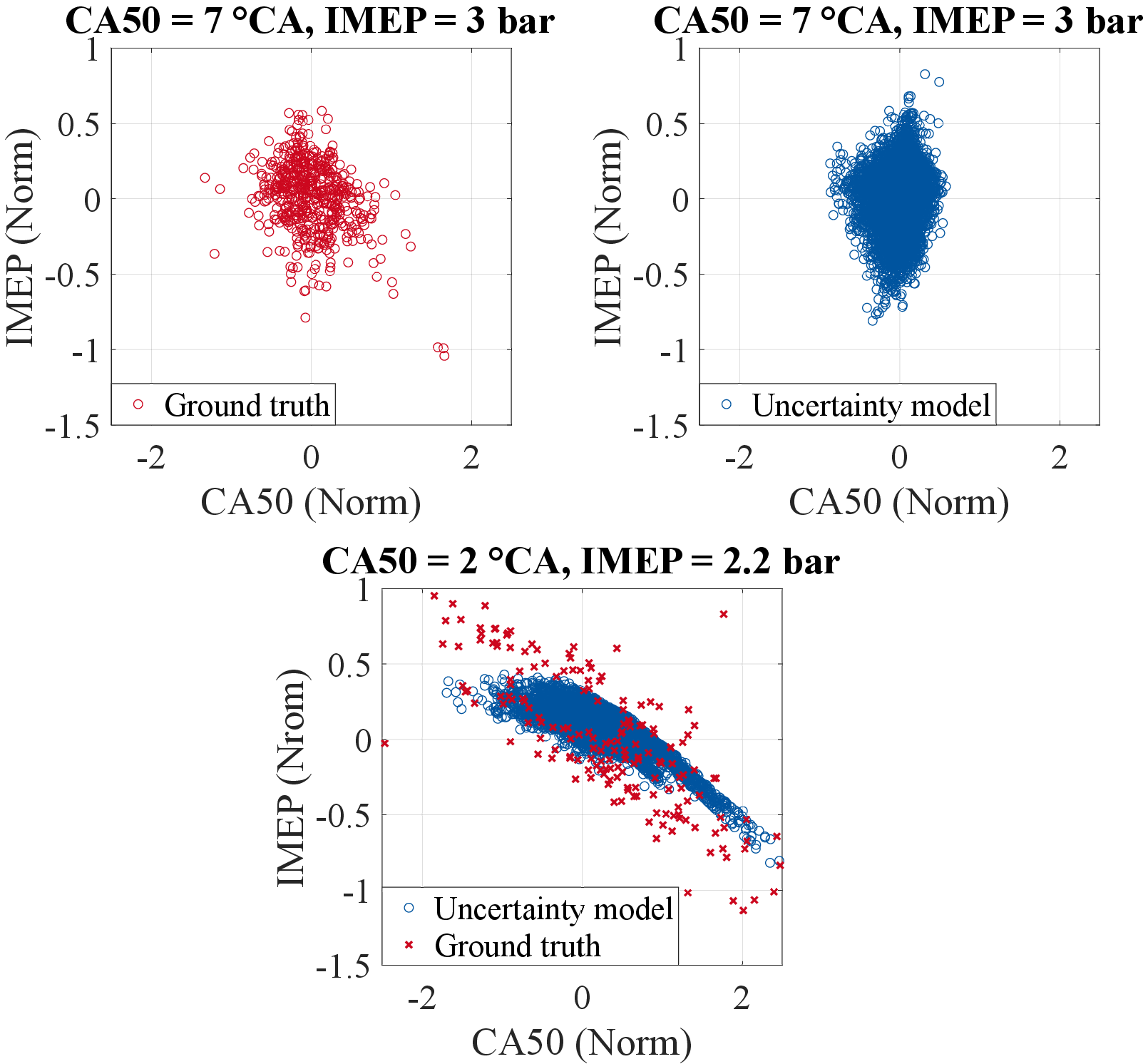}   
\caption{Conditional uncertainty distribution of CA50 and IMEP under two operating conditions: (1) CA50 = \SI{7}{\degreeCA}, IMEP = \SI{3}{\bar}; (2) CA50 = \SI{2}{\degreeCA}, IMEP = \SI{2.2}{\bar}. The red dots represent all observed samples gathered near operating condition (2), representing the empirical conditional probability distribution under this condition.   }  
\label{exp:uncertainty_condition}  
\end{figure}

Next, we provide a comparative evaluation of several control strategies, including nonlinear MPC, Gaussian-based SMPC, PCE-based SMPC, and GEM-SMPC, focusing on two key control performance metrics, namely the stabilization of CA50 and the accuracy of IMEP reference tracking. \par
Figure~\ref{exp:ca50} presents the evolution of CA50 over $120$ engine cycles under these four control strategies with $50$ simulation runs per controller. The mean trajectory of CA50 is shown as a bold solid line, while the shaded region reflects the variability across runs. The dashed line denotes the reference setpoint of \SI{7}{\degreeCA}. We divide the $120$ control cycles into four IMEP phases, each consisting of $30$ cycles under different IMEP conditions with $2.8$, $2.2$, $3.2$, and \SI{3.9}{\bar} shown in Figure~\ref{exp:mean}. \par
We first examine the mean value of CA50 in the closed loops under different controllers. Overall, controllers that incorporate a nonlinear uncertainty model in the prediction phase consistently achieve accurate mean tracking of the CA50 setpoint in Figure~\ref{exp:ca50}. Their average responses remain very close to the reference with only minor deviations. Among these, the GEM-SMPC method demonstrates slightly superior performance. In particular, under the \SI{3.9}{\bar} IMEP condition, standard quadratic-cost-based controllers yield a mean CA50 value that slightly underperforms the setpoint. The controllers without a nonlinear uncertainty model, namely the NMPC and Gaussian-based SMPC, exhibit a clear bias, with the actual tracking mean consistently above the CA50 reference. \par
In terms of variability across the $50$ simulation runs, the shaded regions reveal clear differences in robustness among the control strategies. The GEM-SMPC yields the narrowest uncertainty band, indicating the smallest CA50 variation and strong cycle-to-cycle consistency. The PC-based SMPC exhibits slightly higher variability but still maintains a tightly controlled envelope around the mean trajectory. However, near the IMEP boundary at \SI{3.9}{\bar}, its variation increases noticeably. The Gaussian-based SMPC also produces a relatively narrow band, but it still exhibits considerable variation, particularly in the mid-range stable IMEP condition around \SI{3.2}{\bar}. In contrast, the nonlinear MPC shows the greatest variability, with high cycle-to-cycle fluctuations. \par

\begin{figure}[htb]
\centering
\includegraphics[width=0.48\textwidth]{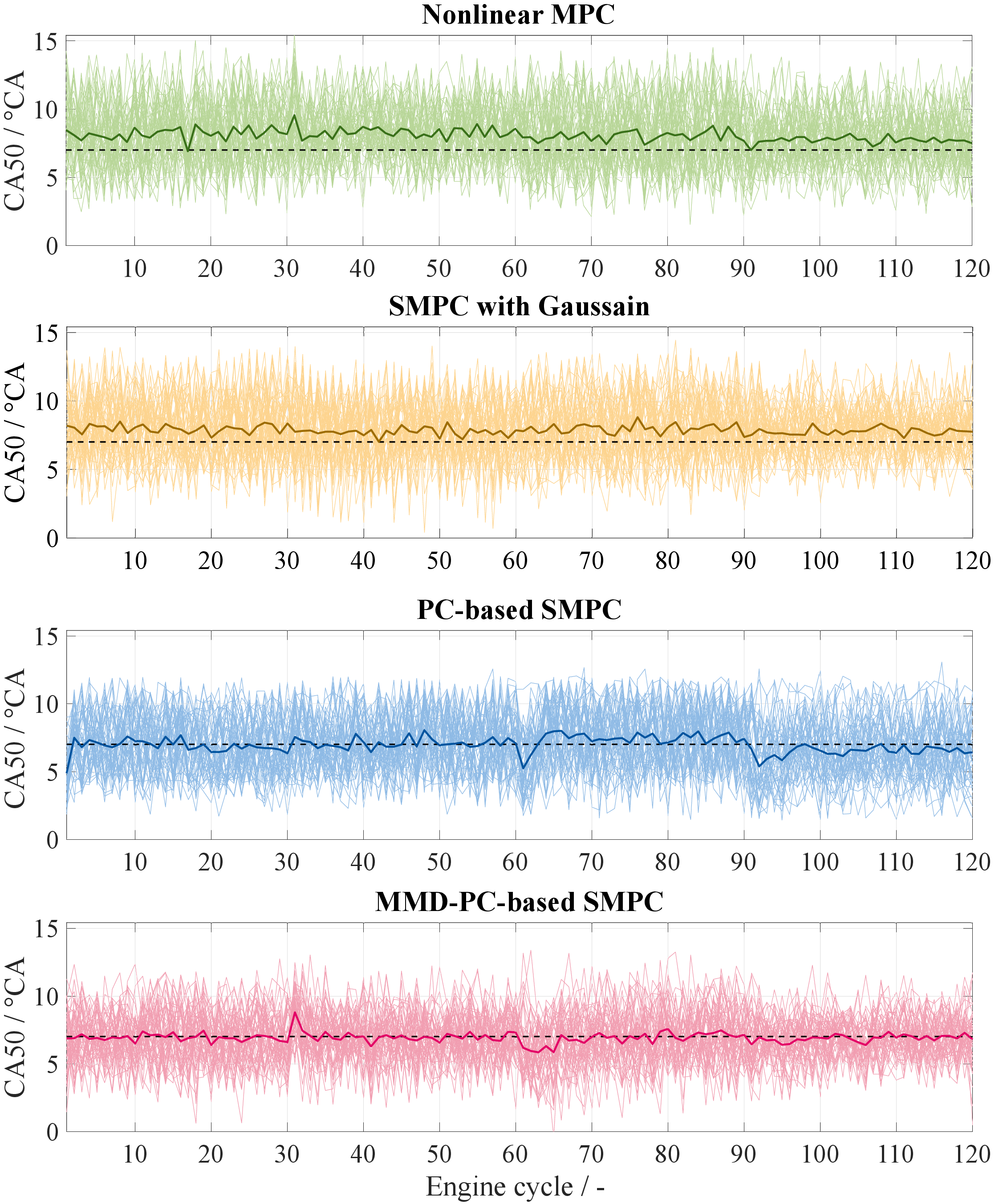}   
\caption{Comparison of CA50 trajectories over $120$ engine cycles under different control strategies: nonlinear MPC, Gaussian-based SMPC, PC-based SMPC, and GEM-SMPC. Each plot shows results from 50 simulation runs. The bold solid line represents the mean CA50, while the dashed line indicates the CA50 set point. The shaded area reflects the variability across runs, highlighting both control performance and combustion stability.  }  
\label{exp:ca50}  
\end{figure}   
\begin{table}[htb]
    \centering
    \caption{Comparison of CA50 Variance Under Different Control Strategies.}
    \begin{tabular}{|l|c|c|c|}
       \hline
         & Variance & Ratio ($\ge13$) &  Ratio ($\le2$)  \\
        \hline
                Nonlinear  & 5.6541  & 2.17 $\%$  & 0.017 $\%$  \\
        \hline
                Gaussian  & 5.1725  & 1.02 $\%$  & 0.32 $\%$  \\
        \hline
                PC  & 5.2215  & 0.017 $\%$  & 0.68 $\%$  \\
        \hline
        GEM  & 3.7025  & 0.1 $\%$  & 0.42 $\%$  \\
        \hline

    \end{tabular}
    \label{teble:variance}
\end{table}

To further quantify the CA50 performance observed in Figure~\ref{exp:ca50}, Table~\ref{teble:variance} summarizes the CA50 variance and the ratios of extreme deviations under different control strategies. Among all methods, GEM-SMPC achieves the lowest variance with $3.7$, indicating the most consistent combustion phasing. The PC-based SMPC and Gaussian-based SMPC have comparable variances, while the nonlinear MPC shows the highest variance with $5.65$, consistent with the wide shaded area observed in Figure~\ref{exp:ca50}. In terms of extreme CA50 constraint violations, both PCE-based methods demonstrate strong performance in maintaining the upper constraint, with violation rates below \SI{0.1}{\percent}. Although the PC-based methods perform slightly worse than the non-PC-based methods in enforcing the lower constraint, they still maintain acceptable performance. When comparing the two PC-based controllers, there is no considerable difference in terms of constraint violations caused by outliers.  \par

\begin{figure}[htb]
\centering
\includegraphics[width=0.48\textwidth]{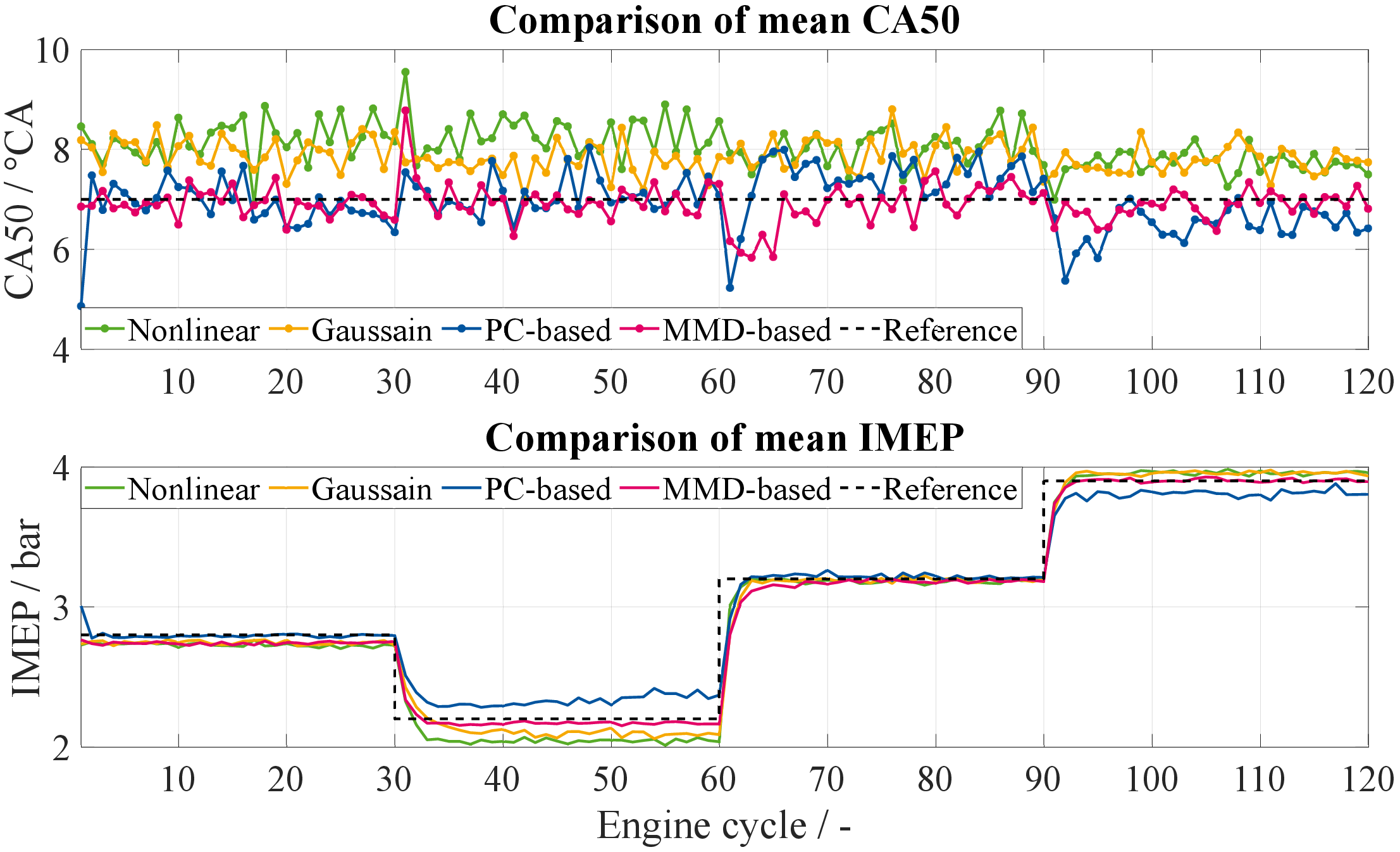}   
\caption{ Comparison of mean CA50 (top) and mean IMEP (bottom) trajectories under different control strategies over $120$ engine cycles, including nonlinear MPC, Gaussian-baed SMPC, PC-based SMPC, and GEM-SMPC. The dashed black line indicates the reference trajectory. Results show differences in combustion phasing accuracy and load tracking performance across the methods. }  
\label{exp:mean}  
\end{figure} 

\begin{table}[htb]
    \centering
    \caption{Comparison of Mean CA50 and IMEP Errors Under Different Control Strategies.}
    \begin{tabular}{|l|c|c|c|c|}
       \hline
               & Nonlinear & Gaussian & PC& GEM   \\
        \hline
       RMSE-CA50 &  1.1410 & 0.9141& 0.5632  & 0.3738     \\
        \hline
       RMSE-IMEP & 0.0933 & 0.0788 & 0.0964    & 0.0580     \\
        \hline
    \end{tabular}
    \label{teble:mean}
\end{table}
Figure~\ref{exp:mean} and Table~\ref{teble:mean} together provide a  comparison of the control accuracy across different strategies in terms of both mean CA50 and mean IMEP. The top subplot of Figure~\ref{exp:mean} shows the same results as in Figure~\ref{exp:ca50}, where the GEM-SMPC consistently tracks the CA50 reference trajectory with minimal deviation, demonstrating superior mean CA50 tracking accuracy. Table~\ref{teble:mean} confirms these observations quantitatively, where GEM-SMPC achieves the lowest CA50 RMSE $0.3738$ with a \SI{34}{\percent} reduction compared to the PC-based MPC and over \SI{60}{\percent} improvement relative to Gaussian-based and Nonlinear MPC. In the bottom subplot of Figure~\ref{exp:mean}, all controllers are generally capable of following the step reference changes in IMEP. However, GEM-SMPC outperforms the others by delivering the smoothest and most accurate tracking of the reference profile across all phases. The other methods, particularly PC-based MPC and nonlinear MPC, show slightly lower tracking accuracy at the boundary conditions, like \SI{2.2}{bar}. Table~\ref{teble:mean} reflects this in the lowest RMSE-IMEP value for GEM-SMPC with $0.058$, followed by Gaussian-based SMPC with $0.0788$. The PC-based MPC and Nonlinear controllers trail behind with higher errors.
\subsection{Discussion}
As shown in Figure~\ref{exp:uncertainty}, Figure~\ref{exp:uncertainty2}, and Figure~\ref{exp:uncertainty_condition}, the learned uncertainty model effectively replicates both marginal and conditional distributions under a variety of operating conditions. It successfully reflects important characteristics such as negative CA50–IMEP correlation and higher-order moment behavior, indicating that the model is capable of approximating complex, real-world uncertainty structures. \par
However, the model underestimates tail behavior and outliers, likely due to insufficient training data exposure, which is a known challenge in data-driven modeling. While this limits the model’s coverage of extreme events, it supports the SMPC design goal of managing constraint violations under uncertainty within a predefined probability threshold while avoiding overly conservative solutions. \par
Moreover, conditional distributions show that the model can adjust to operating regimes, yielding narrow, unimodal distributions in stable regimes and broader, more irregular patterns near boundaries. This adaptability enhances its practical applicability in predictive control under various conditions. \par
As shown in Figure~\ref{exp:ca50}, Figure~\ref{exp:mean}, and the corresponding tables, the control comparison reveals clear benefits of using learned nonlinear uncertainty models. Controllers that include such models, particularly GEM-SMPC, demonstrate improved tracking accuracy and robustness. \par
In terms of CA50 performance, the biases observed in NMPC and Gaussian-based SMPC highlight the limitations of neglecting asymmetry and higher-order characteristics of the underlying uncertainty distribution. These simplifications can lead to systematic control errors under non-trivial uncertainty conditions. The superior performance of the GEM-MPC stems from its use of a non-parametric distance metric that enables more accurate characterization of the true distribution. This supports more precise and reliable control, especially under boundary IMEP conditions where traditional methods degrade. \par
The increased variability of CA50 under the PC-based SMPC near boundary conditions suggests that although effective in stable regimes, this approach could benefit from further refinement in handling distributional shifts in the objective function. Similarly, the Gaussian-based SMPC's limited expressiveness leads to a mismatch of uncertainty impact, reducing the robustness against uncertainty. The nonlinear MPC, with no probabilistic modeling of uncertainty, fails to suppress cycle-to-cycle variations effectively, especially under less stable conditions. The results confirm the importance of stochastic modeling in achieving reliable combustion control under real-world disturbances. \par
In terms of IMEP tracking, the results suggest that the MMD objective also improves IMEP tracking performance, especially under boundary operating conditions. While the pure PCE-based method shows strong capability in enforcing constraints, it exhibits a more conservative response in mean tracking, due to the limitations of the quadratic objective function, which may not adequately reflect nonlinear performance trade-offs. In contrast, the MMD-based formulation directly penalizes distributional deviations, allowing it to better capture nonlinear system behavior and deliver balanced control performance across both constraint holding and load tracking.\par
In summary, GEM-MPC's ability to minimize variance while still respecting constraints highlights its robustness against both advanced and delayed combustion phases. This confirms its potential for improving safety and operational consistency in the presence of nontrivial uncertainty.\par
However, the overall performance of GEM-MPC still depends on the quality and representativeness of the training data used to learn the uncertainty model. In scenarios where data is limited or fails to capture critical events, the GEM-MPC framework may exhibit reduced robustness. Additionally, the sample-based framework incurs a high computational cost, with each GEM-MPC problem requiring an average of $1.25$ seconds to solve, posing a critical barrier to real-time applicability. Future work will explore deployment on real engine platforms by investigating learning-based approaches to approximate GEM-SMPC policies for real-time implementation. 
\section{Conclusion}
This work presents a structured GEM-SMPC framework that integrates a data-driven generative model for conditional residual uncertainty representation, a modified PCE technique for efficient and tractable propagation of additive uncertainty through nonlinear dynamics over the prediction horizon, and an MMD-based objective function for aligning predicted state distributions with desired targets. Together, these components enable robust control of a nonlinear system under complex, state-dependent stochastic uncertainty. \par
We further conduct a comparative study of multiple control strategies for stochastic combustion phasing and load control in the HCCI engine. By explicitly accounting for the effects of nonlinear uncertainty in the prediction and capturing its higher-order distributional characteristics, GEM-MPC improves both combustion stability and tracking accuracy without sacrificing constraint satisfaction compared to the other baseline controllers. Although demonstrated in the context of HCCI combustion, the proposed GEM-SMPC framework is generalizable and transferable to other applications involving stochastic uncertainty.


\section*{Acknowledgments}
The authors gratefully acknowledge funding from the German Research Foundation (DFG) under project number 277012063, within the framework of Research Unit 2401, entitled "Optimization-Based Multiscale Control of Low-Temperature Combustion Engines". Their generous support made this work possible.


%

\bibliographystyle{IEEEtran}
\bibliography{reference_tcst}

\vfill

\end{document}